%% file: main.tex
\pdfoutput=1
\documentclass[conference]{IEEEtran}
\usepackage{amsmath,amssymb}
\usepackage{graphicx}
\usepackage{cite}
\usepackage{tikz}
\usetikzlibrary{arrows.meta,calc,shadows}
\usepackage[hidelinks]{hyperref}

\newcommand{\eff}{\mathrm{eff}}
\newcommand{\SINR}{\gamma}
\newcommand{\E}{\mathbb{E}}
\newtheorem{lemma}{Lemma}
\newtheorem{proposition}{Proposition}
\newtheorem{corollary}{Corollary}

\begin{document}

\title{Impact of Antenna Position Errors on TDMA and NOMA in Pinching-Antenna Systems}

\author{
\IEEEauthorblockN{Wei Jiang}
\IEEEauthorblockA{Intelligent Networking Research Group\\
German Research Center for Artificial Intelligence (DFKI)\\
Kaiserslautern, 67663 Germany}
\and
\IEEEauthorblockN{Hans D. Schotten}
\IEEEauthorblockA{Department of Electrical and Computer Engineering\\
University of Kaiserslautern (RPTU)\\
Kaiserslautern, 67663 Germany}
}

\maketitle

\begin{abstract}
Pinching-antenna systems (PASS) create reconfigurable wireless channels by moving dielectric antennas along a waveguide.
Mechanical position errors translate into phase errors that degrade performance, yet almost all prior PASS studies assume ideal placement. This paper models
this position error, derives the resulting phase error, and analyzes its
impact on time-division multiple access (TDMA) and non-orthogonal
multiple access (NOMA), obtaining their ergodic-rate bounds and
closed-form outage probabilities. Analyses and simulations reveal that
single-antenna TDMA is immune to position error, whereas multi-antenna
TDMA loses array gain and NOMA suffers an interference floor that cannot be overcome by increasing transmit power.
\end{abstract}


\section{Introduction}

\IEEEPARstart{P}{inching}-antenna systems (PASS) couple small dielectric
radiators---\emph{pinching antennas} (PAs), or pinches---onto a
waveguide~\cite{Ref_Ding2025PinchingPerspective,Ref_Liu2025Tutorial}.
Unlike conventional antennas at fixed positions, a PASS creates
reconfigurable wireless channels by moving its antennas, making the
antenna position the optimization variable. This flexibility, however,
has a mechanical price: the position is set by an actuator. Its
resolution, drift, vibration, and thermal creep cause position errors,
which translate into phase errors that degrade performance. Early work on
PASS identifies imprecise antenna positioning as an open
problem~\cite{Ref_Ding2025PinchingPerspective}.

Nevertheless, antenna position error has received little attention: prior
PASS studies almost invariably assume ideal placement, e.g.,
\cite{Ref_Tyrovolas2026Performance,Ref_Yue2026NOMAPerformance}.
So far, only one work~\cite{Ref_Pakravan2026SecrecyPosition} analyzes it, but only as a statistical input to the secrecy
outage of a single-user link, and~\cite{Ref_Chen2026HybridPASS} evaluates its rate
loss by simulation. Two others treat the position as a design variable, without modeling its
error:~\cite{Ref_AlaaEldin2026BERNOMA} shows analytically that, in the
uplink, performance
fluctuates so rapidly with position that, in some regions, a $1$\,mm change
moves the bit error rate by about $10$\,dB;~\cite{Ref_Wang2026RSMAMultiCarrier}
places its antennas by minimizing a phase-shift error and reports from
simulation that rate-splitting multiple access tolerates inaccurate
positions better. The work in~\cite{Ref_Bozkurt2026Trajectory} explicitly asks whether the
actuator error matters and answers that a sub-centimeter displacement
causes a negligible path-loss change, under $0.1$\,dB, but does not
consider the phase.
How a displacement becomes a phase error, and what performance loss that error causes, has not been analyzed in closed form.

To the best of our knowledge, this paper is the first to fill this gap,
with an analytical framework that traces the position error from the
mechanical displacement to the link performance. Specifically, the main
contributions of this paper include:
\begin{enumerate}[\setlength{\itemsep}{0pt}\setlength{\parsep}{0pt}\setlength{\topsep}{1pt}]
\item \emph{Position-error modeling} (Section~\ref{sec:chain}): the
random displacement of a pinch along its waveguide is mapped to the phase
error, its root-mean-square (rms) value, and the coherence factor of the
link.
\item \emph{Performance analysis} (Section~\ref{sec:impact}): for
time-division multiple access (TDMA) and non-orthogonal multiple access
(NOMA), the ergodic rate of each
scheme is bounded, and the outage probability is derived in closed form
for a single pinch under either
scheme and for two-pinch TDMA.
\item \emph{Findings}: analytical and numerical results reveal that
single-pinch TDMA is immune to position error and therefore the most robust;
multi-pinch TDMA degrades gracefully, losing part of its array gain; and
NOMA is the most vulnerable, suffering an interference floor that cannot
be overcome by increasing transmit power.
\end{enumerate}

The remainder of this paper is organized as follows.
Section~\ref{sec:model} presents the system model.
Section~\ref{sec:chain} models the position error, from the displacement to
the phase error and its statistics. Section~\ref{sec:impact} analyzes the
ergodic rate and the outage probability of TDMA and NOMA,
Section~\ref{sec:num} presents numerical results, and
Section~\ref{sec:concl} concludes the paper.

\section{System Model}\label{sec:model}

As shown in Fig.~\ref{fig:sysmodel}, a waveguide along the
$x$-axis at height $d$, fed by a radio-frequency (RF) chain at
$\boldsymbol{\psi}_0=\nobreak(x_0,0,d)$, carries $N$ active PAs at
$\boldsymbol{\psi}_n=\nobreak(x_n,0,d)$, $n=1,\dots,N$, serving $K$
single-antenna users at $\mathbf{u}_k=\nobreak(x_k,y_k,0)$, $k=1,\dots,K$. The
PAs are positioned either by \emph{continuous} placement,
which slides a PA to any point of the waveguide, or by \emph{discrete} placement, in
which PAs are pre-installed at fixed points, spaced by at least
$\lambda/2$, and a subset is
activated~\cite{Ref_Wang2025AntennaActivationNOMA}.
With the feed-to-PA distance
$d_n=\lVert\boldsymbol{\psi}_n-\boldsymbol{\psi}_0\rVert=x_n-x_0$ and the
free-space distance from PA $n$ to user $k$,
$r_{nk}=\lVert\mathbf{u}_k-\boldsymbol{\psi}_n\rVert=\sqrt{(x_k-x_n)^2+y_k^2+d^2}$,
the phase accumulated is
\begin{equation}
\phi_{nk}\triangleq 2\pi r_{nk}/\lambda+2\pi d_n/\lambda_g,
\label{eq:phi}
\end{equation} where $\lambda_g=\lambda/n_\eff$ and $n_\eff>1$ is the effective refractive index.
User $k$ receives
$y_k=\sum_{n=1}^{N}\sqrt{\eta P_n}\,r_{nk}^{-1}e^{-j\phi_{nk}}\,s+w_k$,
with $\eta\triangleq c^2/(4\pi f_c)^2$, $c$ the speed of light, $f_c$ the
carrier frequency, $P_n$ the power radiated by PA $n$, $s$ the unit-power
composite symbol that the single RF chain delivers to every PA, and
$w_k\sim\mathcal{CN}(0,\sigma_w^2)$.
Under the equal-power allocation $P_n=P/N$, with $P$ the total transmit
power, $y_k=\sqrt{P/N}\,g_k\,s+w_k$ with user $k$'s effective channel
\begin{equation}
g_k=\sum\nolimits_{n=1}^{N}g_{nk},\qquad
g_{nk}\triangleq\sqrt{\eta}\,r_{nk}^{-1}e^{-j\phi_{nk}} .
\label{eq:geff}
\end{equation}

\begin{figure}[!t]
\centering
\resizebox{\columnwidth}{!}{\input{fig_sysmodel}}
\caption{System model with antenna position error. The projection of $r_{nk}$ on the waveguide axis gives $\cos\theta=(x_n-x_k)/r_{nk}$.}
\label{fig:sysmodel}
\end{figure}
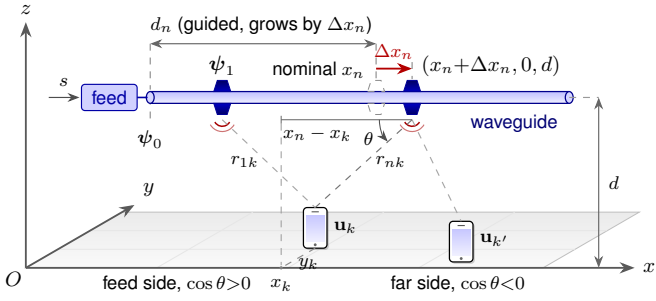

\section{Position-Error Modeling}\label{sec:chain}

Under continuous placement, given a desired position $x_n$, pinch $n$
actually sits at $x_n+\Delta x_n$, where $\Delta x_n$ is the mechanical
mismatch.\footnote{Under discrete placement, the pinch can be placed only
at pre-installed positions spaced by $L$, so the desired $x_n$ is generally
unavailable: the nearest pre-installed position, $x_n+q_n$ with the
quantization offset $|q_n|\le L/2$, is activated instead, and the pinch
actually sits at $x_n+q_n+\Delta x_n$. Unlike $\Delta x_n$, $q_n$ is deterministic
and known, so it can be accounted for in the design; what remains random
and unknown is $\Delta x_n$, to which our analysis applies
straightforwardly.} Assuming the
$\Delta x_n$, $n=1,\dots,N$, are independent and identically distributed (i.i.d.),
this section models the phase error for one antenna and drops the index
$n$.

\subsection{Phase Sensitivity: From $\Delta x$ to $\Delta\phi$}
Both terms of \eqref{eq:phi} depend on $x_n$, so a displacement $\Delta x$
produces a phase error $\Delta\phi$.

\begin{lemma}\label{lem:S}
For $\Delta x\ll r_{nk}$, the phase error is linear in the displacement,
i.e.,
\begin{equation}
\Delta\phi=S\,\Delta x,\qquad
S=\frac{2\pi}{\lambda}\big(n_\eff+\cos\theta\big),
\label{eq:dphi}
\end{equation}
where we define the phase sensitivity $S$ as the phase change per unit
displacement, with $\theta$ as in Fig.~\ref{fig:sysmodel}.
\end{lemma}
\begin{IEEEproof}
By \eqref{eq:phi}, $\Delta\phi=(2\pi/\lambda)\,\Delta r_{nk}
+(2\pi/\lambda_g)\,\Delta d_n$. The feed-to-PA distance $d_n=x_n-x_0$ is
linear in $x_n$, so $\Delta d_n=\Delta x$ exactly. The free-space distance
is not: let $\tilde r_{nk}$ denote the actual distance, with the pinch at
$x_n+\Delta x$. Then
\begin{equation}
\begin{split}
\Delta r_{nk}&\triangleq\tilde r_{nk}-r_{nk}\\
&=\sqrt{(x_k-x_n-\Delta x)^2+y_k^2+d^2}-r_{nk}\\
&=\sqrt{r_{nk}^2+2r_{nk}\cos\theta\,\Delta x+\Delta x^2}-r_{nk}\\
&=r_{nk}\Big(\sqrt{1+2\cos\theta\,\Delta x/r_{nk}
+\Delta x^2/r_{nk}^2}-1\Big).
\end{split}
\label{eq:dr}
\end{equation}
The root in \eqref{eq:dr} expands to second order as
$1+\cos\theta\,\Delta x/r_{nk}
+\tfrac{1}{2}\sin^2\!\theta\,\Delta x^2/r_{nk}^2$, so that
\begin{equation}
\Delta r_{nk}=\cos\theta\,\Delta x
+\sin^2\!\theta\,\frac{\Delta x^2}{2r_{nk}}
+O\!\Big(\frac{\Delta x^3}{r_{nk}^2}\Big).
\label{eq:drexp}
\end{equation}
Substituting $\Delta d_n$ and $\Delta r_{nk}$ in \eqref{eq:phi} gives
\begin{equation}
\Delta\phi=S\,\Delta x
+\frac{\pi}{\lambda}\sin^2\!\theta\,\frac{\Delta x^2}{r_{nk}}
+O\!\Big(\frac{\Delta x^3}{\lambda r_{nk}^2}\Big).
\label{eq:dphiexp}
\end{equation}
For any $\theta$, the second term of \eqref{eq:dphiexp}, as a fraction of
the first, is
\begin{equation*}
\begin{split}
\frac{|\Delta\phi-S\,\Delta x|}{S\,\Delta x}
&=\frac{\sin^2\!\theta}{2(n_\eff+\cos\theta)}\,\frac{\Delta x}{r_{nk}}\\
&\le\Big(n_\eff-\sqrt{n_\eff^2-1}\Big)\frac{\Delta x}{r_{nk}}
<\frac{\Delta x}{r_{nk}} ,
\end{split}
\end{equation*}
$\lambda$ cancelling and higher orders contributing a further factor
$\Delta x/r_{nk}$; over $\cos\theta\in[-1,1]$ the geometric factor peaks
at $\cos\theta=-(n_\eff-\sqrt{n_\eff^2-1})$, where it equals that same
value, below $1$ for $n_\eff>1$. With $\Delta x$ in millimeters and
$r_{nk}$ in meters, $\Delta x/r_{nk}$ is of order $10^{-3}$, so the
remainder is negligible and \eqref{eq:dphi} follows.
\end{IEEEproof}

\begin{corollary}[Sensitivity bounds]\label{cor:S}
The sensitivity of \eqref{eq:dphi} obeys
\begin{equation}
0<\frac{2\pi}{\lambda}\big(n_\eff-1\big)\le S\le S_{\max}
\triangleq\frac{2\pi}{\lambda}\big(n_\eff+1\big).
\label{eq:Sbounds}
\end{equation}
\end{corollary}
\begin{IEEEproof}
Apply $|\cos\theta|\le1$ to $S=(2\pi/\lambda)(n_\eff+\cos\theta)$; the
lower bound on the left is positive because $n_\eff>1$.
\end{IEEEproof}

\subsection{Statistical Characterization: From $\Delta\phi$ to $\chi$}\label{sec:stat}
The mismatch $\Delta x$ is random and unknown, so performance can be
judged only through its deterministic rms value
$\sigma_x\triangleq\sqrt{\E[\Delta x^2]}$, and that of the phase error,
$\sigma_\phi\triangleq\sqrt{\E[\Delta\phi^2]}$. By
Lemma~\ref{lem:S}, $\Delta\phi=S\Delta x$ with $S$ fixed by the
geometry, so $\E[\Delta\phi^2]=S^2\E[\Delta x^2]$ and, for any
distribution of $\Delta x$,
\begin{equation}
\sigma_\phi=S\,\sigma_x\;\le\;S_{\max}\,\sigma_x ,
\label{eq:rms}
\end{equation}
the bound following from \eqref{eq:Sbounds}.
A position error rotates the per-pinch channel $g_{nk}$ of \eqref{eq:geff},
written $g$ here, into $\tilde g=g\,e^{j\Delta\phi}$. Its amplitude changes
by the relative amount $\Delta x\cos\theta/r_{nk}$, which is negligible, so
$|\tilde g|=|g|$: no realization loses power. What is lost on averaging is
coherence, which depends on the distribution of $\Delta\phi$. Being a sum
of many small independent contributions, $\Delta x$, hence
$\Delta\phi=S\Delta x$, is taken Gaussian by the central limit theorem, and
zero-mean after calibration; by symmetry, the sign of the exponent in
$\tilde g$ does not matter.

\begin{lemma}[Coherence factor]\label{lem:rms}
For zero-mean Gaussian $\Delta\phi\sim\mathcal N(0,\sigma_\phi^2)$,
\begin{equation}
\chi\triangleq\E\big[e^{j\Delta\phi}\big]=e^{-\sigma_\phi^2/2}\in(0,1],
\label{eq:chi}
\end{equation}
so that the mean channel is $\E[\tilde g]=\chi g$.
\end{lemma}
\begin{IEEEproof}
The Gaussian density of $\Delta\phi$ is
$e^{-\phi^2/2\sigma_\phi^2}/\sqrt{2\pi\sigma_\phi^2}$, so
$\E[e^{j\Delta\phi}]=\int_{-\infty}^{\infty}
e^{\,j\phi-\phi^2/2\sigma_\phi^2}/\sqrt{2\pi\sigma_\phi^2}\,d\phi$. Since
$j\phi-\phi^2/2\sigma_\phi^2=-(\phi-j\sigma_\phi^2)^2/2\sigma_\phi^2-\sigma_\phi^2/2$,
we have
\begin{equation*}
\E\big[e^{j\Delta\phi}\big]=e^{-\sigma_\phi^2/2}\int_{-\infty}^{\infty}
\frac{e^{-(\phi-j\sigma_\phi^2)^2/2\sigma_\phi^2}}{\sqrt{2\pi\sigma_\phi^2}}\,d\phi
=e^{-\sigma_\phi^2/2},
\end{equation*}
the last integral being a Gaussian density shifted by an imaginary
constant, which integrates to one. Since $e^{-\sigma_\phi^2/2}$ decreases
monotonically in $\sigma_\phi\ge0$, $\chi=1$ at $\sigma_\phi=0$, ideal
placement, and $\chi\to0$ as $\sigma_\phi\to\infty$: hence the range in
\eqref{eq:chi}.
\end{IEEEproof}

\section{Performance Analysis}\label{sec:impact}

This section analyzes the performance impact of the position error,
bounding the ergodic rate and deriving the closed-form outage probability
of TDMA and NOMA.

\subsection{Ergodic Rate}
\subsubsection{TDMA}
User $k$ is served alone in its slot by all $N$ pinches, with
$|x_n-x_{n'}|\ll r_{nk}$, so that $r_{nk}\approx r_{n'k}$: their channels
$g_{nk}$ then share one gain,
$|g_{nk}|=\beta_k$, $\forall n$, and one sensitivity $S$, and their desired positions give
$\phi_{nk}\equiv\phi_{n'k}\ (\mathrm{mod}\ 2\pi)$, so the $N$ contributions
add coherently~\cite{Ref_Ding2025PinchingPerspective}.

\begin{proposition}[Ergodic rate of TDMA]\label{prop:aN}
For i.i.d. $\Delta\phi_{nk}\sim\mathcal N(0,\sigma_\phi^2)$, the ergodic rate
$\bar R_k\triangleq\E[R_k]$ is bounded as
\begin{equation}
\bar R_k\le\log_2\!\left(1+\frac{NP\beta_k^2}{\sigma_w^2}
\Big[\chi^2+\frac{1-\chi^2}{N}\Big]\right).
\label{eq:aN}
\end{equation}
\end{proposition}
\begin{IEEEproof}
Considering $\Delta\phi_{nk}$, \eqref{eq:geff} gives the instantaneous rate as
\begin{equation}
R_k=\log_2\!\left(1+\frac{P}{N\sigma_w^2}\Big|\sum\nolimits_{n=1}^{N}g_{nk}e^{j\Delta\phi_{nk}}\Big|^2\right).
\label{eq:tdma}
\end{equation}
Define the array gain
$a(N)\triangleq\big|\sum_ng_{nk}e^{j\Delta\phi_{nk}}\big|^2/\sum_n|g_{nk}|^2$ and the
error-free signal-to-noise ratio (SNR) $\gamma_0\triangleq(P/N\sigma_w^2)|\sum_ng_{nk}|^2$. As $\sum_n|g_{nk}|^2=N\beta_k^2$ and $|\sum_ng_{nk}|^2=N^2\beta_k^2$,
$\gamma_0=NP\beta_k^2/\sigma_w^2$ and
\eqref{eq:tdma} becomes
$R_k=\log_2(1+\gamma_0\,a(N)/N)$. Since $\log_2(1+x)$ is concave,
Jensen's inequality gives
$\bar R_k\le\log_2(1+\gamma_0\,\E[a(N)]/N)$. By independence and \eqref{eq:chi}, and with
$\sum_{n\ne n'}g_{nk}g_{n'k}^*=|\sum_ng_{nk}|^2-\sum_n|g_{nk}|^2=N(N-1)\beta_k^2$,
we have
\begin{equation}
\begin{split}
\E[a(N)]&=\frac{1}{N\beta_k^2}\sum_{n}\sum_{n'}g_{nk}g_{n'k}^*\,
\E\big[e^{j(\Delta\phi_{nk}-\Delta\phi_{n'k})}\big]\\
&=\frac{1}{N\beta_k^2}\Big(\sum\nolimits_n|g_{nk}|^2
+\sum\nolimits_{n\ne n'}g_{nk}g_{n'k}^*\,\chi^2\Big)\\
&=\frac{N\beta_k^2+N(N-1)\beta_k^2\chi^2}{N\beta_k^2}=1+(N-1)\chi^2 ,
\end{split}
\label{eq:gain}
\end{equation}
which gives \eqref{eq:aN}.
\end{IEEEproof}

\begin{corollary}[Immunity of single-pinch TDMA]\label{cor:immune}
For $N=1$, \eqref{eq:tdma} is rewritten as
\begin{equation*}
R_k=\log_2\!\big(1+\tfrac{P}{\sigma_w^2}|g_{1k}e^{j\Delta\phi_{1k}}|^2\big)
=\log_2\!\big(1+\tfrac{P\beta_k^2}{\sigma_w^2}\big),
\quad\forall\Delta\phi_{1k},
\end{equation*}
and $\bar R_k=\log_2(1+P\beta_k^2/\sigma_w^2)$. So \eqref{eq:aN} holds with equality,
as $\chi^2+(1-\chi^2)/N=1$, and the effect of phase error vanishes. Under the
line-of-sight channel, a single pinch is therefore immune to position
error: its rate and outage probability are exactly those of ideal
placement, as analyzed
in~\cite{Ref_Ding2025PinchingPerspective,Ref_Tyrovolas2026Performance}.
\end{corollary}

\subsubsection{NOMA}
The signals of $K$ users are superposed with powers $\alpha_jP$,
$\sum_j\alpha_j=1$. User $k$ cancels those of the weaker users $j<k$ by
successive interference cancellation (SIC) with the replica $g_k=\sum_ng_{nk}$
of \eqref{eq:geff}. As its actual channel is
$\tilde g_k=\sum_ng_{nk}e^{j\Delta\phi_{nk}}$, its signal-to-interference-plus-noise
ratio (SINR) becomes
\begin{equation}
\SINR_k=\frac{\alpha_k|\tilde g_k|^2}{N\sigma_w^2/P+|\tilde g_k|^2\sum_{j>k}\alpha_j
+|\tilde g_k-g_k|^2\sum_{j<k}\alpha_j} .
\label{eq:sinrK}
\end{equation}
The position error enters \eqref{eq:sinrK} only through $|\tilde g_k|$
and $|\tilde g_k-g_k|$, which do not depend on $K$. It therefore suffices
to consider a pair, $K=2$, of a strong user $s$ and a weak user $w$
with $\alpha_s+\alpha_w=1$, $\alpha_s<\alpha_w$. Its weak user has
$\SINR_w=\alpha_w|\tilde g_w|^2/(N\sigma_w^2/P+\alpha_s|\tilde g_w|^2)$,
where the error enters only through
$|\tilde g_w|^2=|\sum_ng_{nw}e^{j\Delta\phi_{nw}}|^2$, the factor already treated
for TDMA in \eqref{eq:tdma} and \eqref{eq:gain}. Therefore, we only analyze the impact of position
errors on the strong user. For $N\ge2$, the strong user's signal
$|\tilde g_s|^2=|\sum_ng_{ns}e^{j\Delta\phi_{ns}}|^2$ and residual
$|\tilde g_s-g_s|^2=|\sum_ng_{ns}(e^{j\Delta\phi_{ns}}-1)|^2$ are both random due
to the $N$ phase errors. Its ergodic rate is therefore analytically
intractable.
In contrast, at $N=1$ the signal $|g_se^{j\Delta\phi_{1s}}|^2=|g_s|^2$ is not
random, so we derive a closed-form bound for this case. It is also
practical, since pairing two far-apart users lets one pinch dominate each
user's channel~\cite{Ref_Ding2025PinchingPerspective}.

\begin{proposition}[Ergodic rate of NOMA]\label{prop:sic}
For a single pinch and $\Delta\phi_{1s}\sim\mathcal N(0,\sigma_\phi^2)$, the
ergodic rate $\bar R_s\triangleq\E[R_s]$ of the strong user is bounded as
\begin{equation}
\bar R_s\ge\log_2\!\Big(1+\frac{\alpha_sP|g_s|^2}
{\sigma_w^2+2(1-\chi)\alpha_wP|g_s|^2}\Big).
\label{eq:floor}
\end{equation}
\end{proposition}
\begin{IEEEproof}
At $K=2$ and $N=1$, $\tilde g_s=g_se^{j\Delta\phi_{1s}}$ gives
$|\tilde g_s|=|g_s|$ and $|\tilde g_s-g_s|^2=|g_s|^2\Omega$, where
\begin{equation}
\Omega\triangleq\big|e^{j\Delta\phi_{1s}}-1\big|^2=2-2\cos\Delta\phi_{1s}
=4\sin^2\!\frac{\Delta\phi_{1s}}{2}
\label{eq:dbar}
\end{equation}
is the residual fraction of the cancelled power, zero under ideal
placement. Then \eqref{eq:sinrK} becomes
\begin{equation}
\SINR_s=\frac{\alpha_sP|g_s|^2}{\sigma_w^2+\alpha_wP|g_s|^2\,\Omega},
\qquad R_s=\log_2(1+\SINR_s).
\label{eq:noma}
\end{equation}
By \eqref{eq:dbar} and \eqref{eq:chi},
$\bar\Omega\triangleq\E[\Omega]=2-2\,\E[\cos\Delta\phi_{1s}]=2(1-\chi)$.
With $u\triangleq\sigma_w^2+\alpha_wP|g_s|^2\Omega$,
$\partial^2R_s/\partial\Omega^2\propto u^{-2}-(u+\alpha_sP|g_s|^2)^{-2}>0$,
so $R_s$ is convex in $\Omega$. Jensen's inequality then gives
$\bar R_s\ge\log_2\!\big(1+\alpha_sP|g_s|^2/(\sigma_w^2+\alpha_wP|g_s|^2\bar\Omega)\big)$,
which is \eqref{eq:floor}.
\end{IEEEproof}

\subsection{Outage Probability}\label{sec:outage}
We follow the setting
of~\cite{Ref_Ding2025PinchingPerspective,Ref_Tyrovolas2026Performance},
where the user is uniform in a square of side $D$, and the serving pinches
sit above it, $x_n\approx x_k$. Hence
$\cos\theta_{nk}=0$, $S=2\pi n_\eff/\lambda$ is deterministic, and
$\beta_k^2=\eta/r_{nk}^2=\eta/(y^2+d^2)$. The error-free SNR
$\gamma_0=NP\beta_k^2/\sigma_w^2$ is then a function of $y$, i.e.,
$\gamma_0(y)=NC_0/(y^2+d^2)$, $C_0\triangleq\eta P/\sigma_w^2$. A
closed-form outage exists only when it is decided by a single
variable, whose tail probability is the Gaussian tail function
$Q(\cdot)$: for $N\le2$ under TDMA,
whose SNR depends on the phase difference only, and for $N=1$ under NOMA,
whose residual depends on the single phase.

\subsubsection{TDMA}
At $N=1$, by Corollary~\ref{cor:immune}, the outage analysis of ideal
placement applies for every $\sigma_x$.
For a target SINR $\SINR_{\rm th}$, the outage probability
is $P_{\rm out}=1-(2/D)\sqrt{C_0/\SINR_{\rm th}-d^2}$
for
$d^2\le C_0/\SINR_{\rm th}\le d^2+D^2/4$~\cite[Prop.~2]{Ref_Tyrovolas2026Performance}.

\begin{proposition}[Outage of TDMA]\label{prop:tdma_out}
For $N=2$ pinches with i.i.d. $\Delta\phi_{nk}\sim\mathcal N(0,\sigma_\phi^2)$,
the conditional outage probability is
\begin{equation}
P_{\rm out}(y)=\begin{cases}
2Q\Big(\dfrac{\phi_\rho}{\sqrt2\,\sigma_\phi}\Big)-2Q\Big(\dfrac{2\pi-\phi_\rho}{\sqrt2\,\sigma_\phi}\Big)+\epsilon, & \rho<1,\\
1, & \rho\ge1,
\end{cases}
\label{eq:pout_tdma2}
\end{equation}
with $\rho\triangleq\SINR_{\rm th}/\gamma_0(y)$,
$\phi_\rho\triangleq\arccos(2\rho-1)$,
$0\le\epsilon\le2Q((2\pi+\phi_\rho)/(\sqrt2\,\sigma_\phi))$ and
$Q(x)\triangleq\frac{1}{\sqrt{2\pi}}\int_x^\infty e^{-t^2/2}\,dt$.
\end{proposition}
\begin{IEEEproof}
Since
$\Delta\phi_{1k},\Delta\phi_{2k}
\overset{\mathrm{i.i.d.}}{\sim}\mathcal N(0,\sigma_\phi^2)$,
their difference
$\zeta\triangleq\Delta\phi_{1k}-\Delta\phi_{2k}$
follows $\mathcal N(0,2\sigma_\phi^2)$.
Using $g_{1k}=g_{2k}$ and factoring out
$g_{1k}e^{j\Delta\phi_{2k}}$ from \eqref{eq:tdma} at $N=2$ gives
\[
\gamma(y)
=\frac{P\beta_k^2}{2\sigma_w^2}|1+e^{j\zeta}|^2
=\frac{\gamma_0(y)}{2}(1+\cos\zeta),
\]
where $\gamma_0(y)=2P\beta_k^2/\sigma_w^2$ and
$|1+e^{j\zeta}|^2=2(1+\cos\zeta)$. Therefore,
\[
P_{\rm out}(y)
\triangleq\Pr[\gamma(y)<\gamma_{\rm th}]
=\Pr[\cos\zeta<2\rho-1].
\]
For $\rho>1$, $\cos\zeta\le1<2\rho-1$ holds surely;
for $\rho=1$, $\cos\zeta<1$ fails only at
$\zeta\in2\pi\mathbb Z$, a set of probability zero.
Thus, $P_{\rm out}(y)=1$ for $\rho\ge1$,
the second case of \eqref{eq:pout_tdma2}.

For $0<\rho<1$, $\phi_\rho\in(0,\pi)$, and on $[0,2\pi]$ outage occurs
precisely when $\phi_\rho<\zeta<2\pi-\phi_\rho$.
Periodicity generates the positive outage intervals by shifts of
$2\pi m$, $m\ge0$; their negative counterparts have equal
probabilities by symmetry. Hence,
\[
\begin{aligned}
P_{\rm out}(y)
&=2\sum_{m=0}^{\infty}
\Pr\big[
2\pi m+\phi_\rho<\zeta<2\pi(m+1)-\phi_\rho
\big]\\
&=2\sum_{m=0}^{\infty}
\big[
\Pr[\zeta>2\pi m+\phi_\rho]\\
&\qquad\qquad
-\Pr[\zeta>2\pi(m+1)-\phi_\rho]
\big]\\
&=2\sum_{m=0}^{\infty}
\Big[
Q\Big(\frac{2\pi m+\phi_\rho}{\sqrt2\,\sigma_\phi}\Big)
-
Q\Big(\frac{2\pi(m+1)-\phi_\rho}{\sqrt2\,\sigma_\phi}\Big)
\Big],
\end{aligned}
\]
where the last equality follows from
$\zeta\sim\mathcal N(0,2\sigma_\phi^2)$.

Separating the $m=0$ term yields the first case of \eqref{eq:pout_tdma2},
with $\epsilon$ collecting the remaining intervals.
These intervals and their negative counterparts lie in
$\{|\zeta|>2\pi+\phi_\rho\}$; therefore,
\[
0\le\epsilon
\le\Pr[|\zeta|>2\pi+\phi_\rho]
=2Q\Big(\frac{2\pi+\phi_\rho}{\sqrt2\,\sigma_\phi}\Big).\IEEEQEDhereeqn
\]
\end{IEEEproof}

Over $y$ uniform on $[-D/2,D/2]$, the outage probability is
$P_{\rm out}=\frac1D\int_{-D/2}^{D/2}P_{\rm out}(y)\,dy=\frac2D\int_0^{D/2}P_{\rm out}(y)\,dy$,
as $P_{\rm out}(y)$ is even in $y$; at $N=2$ it has no closed form and is
evaluated numerically.

\subsubsection{NOMA}

For $N=1$, with $P|g_s|^2/\sigma_w^2=\gamma_0(y)$, \eqref{eq:noma} becomes
\begin{equation}
\gamma_s(y)
=\frac{\alpha_s\gamma_0(y)}
{1+\alpha_w\gamma_0(y)\Omega},
\qquad
\gamma_0(y)=\frac{C_0}{y^2+d^2}.
\label{eq:sinr_real}
\end{equation}

\begin{proposition}[Outage of NOMA]\label{prop:noma_out}
Let $\Delta\phi_{1s}\sim\mathcal N(0,\sigma_\phi^2)$, with
$\sigma_\phi>0$.

(a) Define the residual tolerance
\begin{equation}
c(y)\triangleq
\frac{\alpha_s}{\gamma_{\rm th}\alpha_w}
-\frac{y^2+d^2}{\alpha_w C_0}.
\label{eq:c}
\end{equation}
Then $P_{\rm out}(y)=1$ for $c(y)\le0$, and
$P_{\rm out}(y)=0$ for $c(y)\ge4$.
For $0<c(y)<4$, write $c=c(y)$ and define
$\phi_c\triangleq2\arcsin(\sqrt c/2)\in(0,\pi)$.
The conditional outage satisfies
\begin{equation}
\begin{aligned}
P_{\rm out}(y)
&=2Q\!\left(\frac{\phi_c}{\sigma_\phi}\right)
-2Q\!\left(\frac{2\pi-\phi_c}{\sigma_\phi}\right)
+\epsilon,\\
0\le\epsilon
&\le2Q\!\left(\frac{2\pi+\phi_c}{\sigma_\phi}\right).
\end{aligned}
\label{eq:pout_cond}
\end{equation}

(b) When $c(0)>0$, define
$Y_1=\sqrt{\alpha_s C_0/\gamma_{\rm th}-d^2}$, so that
$c(y)=(Y_1^2-y^2)/(\alpha_wC_0)>0$ exactly for $|y|<Y_1$.
For $Y_1\le D/2$ and $y\sim{\rm Unif}[-D/2,D/2]$,
the spatially averaged outage is bounded by
\begin{equation}
\begin{aligned}
P_{\rm out}
&\le1-\frac{2Y_1}{D}F(\kappa),\qquad
\kappa
\triangleq\frac{Y_1^2}{2\sigma_\phi^2\alpha_w C_0}
=\frac{c(0)}{2\sigma_\phi^2},\\
F(\kappa)
&\triangleq
\sqrt{\frac{\pi\kappa}{4}}\,e^{-\kappa/2}
\left[
I_0\!\left(\frac{\kappa}{2}\right)
+I_1\!\left(\frac{\kappa}{2}\right)
\right],
\end{aligned}
\label{eq:pout_avg}
\end{equation}
where $I_\nu$ is the modified Bessel function of the first kind
and $0<F(\kappa)<1$.
\end{proposition}

\begin{IEEEproof}
(a) From \eqref{eq:sinr_real},
\begin{equation*}
\begin{aligned}
P_{\rm out}(y)
&=\Pr[\gamma_s(y)<\gamma_{\rm th}]\\
&=\Pr\!\left[
\Omega>
\frac{\alpha_s}{\gamma_{\rm th}\alpha_w}
-\frac{1}{\alpha_w\gamma_0(y)}
\right]\\
&=\Pr[\Omega>c(y)].
\end{aligned}
\end{equation*}
The second equality multiplies by the positive denominator and
divides by $\gamma_{\rm th}\alpha_w\gamma_0(y)>0$.
Since $0\le\Omega\le4$ by \eqref{eq:dbar} and $\Pr[\Omega=0]=0$,
the cases $c\le0$ and $c\ge4$ follow.
For $0<c<4$,
\begin{equation*}
\Omega>c
\quad\Longleftrightarrow\quad
\cos\Delta\phi_{1s}<1-\frac c2=\cos\phi_c.
\end{equation*}
This gives the same periodic outage intervals as in
Proposition~\ref{prop:tdma_out}.
Applying its Gaussian-tail calculation with $\phi_c$ and
$\sigma_\phi$ in place of $\phi_\rho$ and
$\sqrt2\,\sigma_\phi$, respectively, yields
\eqref{eq:pout_cond}.

(b) Since $|\sin u|\le|u|$,
\begin{equation*}
\Omega=4\sin^2(\Delta\phi_{1s}/2)
\le(\Delta\phi_{1s})^2.
\end{equation*}
Thus, for $0\le y\le Y_1$, where
$c(y)=(Y_1^2-y^2)/(\alpha_w C_0)\ge0$,
\begin{equation*}
\begin{aligned}
P_{\rm out}(y)
&\le\Pr[(\Delta\phi_{1s})^2>c(y)]\\
&=2Q\!\left(\frac{\sqrt{c(y)}}{\sigma_\phi}\right)
=2Q\!\left(\frac{\sqrt{Y_1^2-y^2}}{v}\right),
\end{aligned}
\end{equation*}
with $v\triangleq\sigma_\phi\sqrt{\alpha_w C_0}$.
For $Y_1<y\le D/2$, $c(y)<0$ and $P_{\rm out}(y)=1$.
Averaging over $y$ therefore gives
\begin{equation*}
\begin{aligned}
P_{\rm out}
&\le1-\frac{2Y_1}{D}+\frac{2J}{D},\\
J&\triangleq\int_0^{Y_1}
2Q\!\left(\frac{\sqrt{Y_1^2-y^2}}{v}\right)\,dy.
\end{aligned}
\end{equation*}
Integrating by parts, with $2Q(0)=1$ and
$Q'(x)=-e^{-x^2/2}/\sqrt{2\pi}$, and then substituting $y=Y_1\sin t$,
with $\kappa=Y_1^2/(2v^2)$, yields
\begin{equation*}
\begin{aligned}
J
&=Y_1-\frac{2}{v\sqrt{2\pi}}
\int_0^{Y_1}
\frac{y^2e^{-(Y_1^2-y^2)/(2v^2)}}
{\sqrt{Y_1^2-y^2}}\,dy\\
&=Y_1-\frac{2Y_1^2}{v\sqrt{2\pi}}
\int_0^{\pi/2}\sin^2t\,e^{-\kappa\cos^2t}\,dt.
\end{aligned}
\end{equation*}
With $u=\pi-2t$, so that $2\sin^2t=1+\cos u$ and
$2\cos^2t=1-\cos u$, the remaining integral becomes
\begin{equation*}
\begin{aligned}
\int_0^{\pi/2}\sin^2t\,e^{-\kappa\cos^2t}\,dt
&=\frac{e^{-\kappa/2}}4
\int_0^\pi(1+\cos u)e^{(\kappa/2)\cos u}\,du\\
&=\frac{\pi e^{-\kappa/2}}4
\left[
I_0\!\left(\frac\kappa2\right)
+I_1\!\left(\frac\kappa2\right)
\right],
\end{aligned}
\end{equation*}
where the last equality follows from the integral
representations of $I_0$ and $I_1$,
$\pi I_\nu(z)=\int_0^\pi e^{z\cos u}\cos(\nu u)\,du$.
Consequently, with $Y_1/v=\sqrt{2\kappa}$, $J=Y_1[1-F(\kappa)]$, which gives
\eqref{eq:pout_avg}.
Moreover, $0<J<Y_1$, as $0<2Q(x)<1$ for $x>0$, implies $0<F(\kappa)<1$.
\end{IEEEproof}

For $c(0)\le0$, every location is in outage and
$P_{\rm out}=1$.
For $0<Y_1\le D/2$, $F(\kappa)\to1$ as
$\sigma_\phi\to0$, as $\kappa\to\infty$ and $I_\nu(z)\sim e^z/\sqrt{2\pi z}$, so \eqref{eq:pout_avg} recovers the
ideal-placement outage $1-2Y_1/D$, which is
\cite[Prop.~2]{Ref_Tyrovolas2026Performance} with $C_0$ replaced by
$\alpha_sC_0$.
Thus, the served half-width shrinks from $Y_1$ to no less than
$Y_1F(\kappa)$. As $\sigma_\phi\to\infty$, $\kappa\to0$ and $F(\kappa)\to0$.
The parameter $\kappa=c(0)/(2\sigma_\phi^2)$ compares the
maximum residual tolerance with the phase-error variance.

For $Y_1>D/2$, where the target is attainable everywhere, the spatial outage satisfies
\begin{equation*}
P_{\rm out}(0)
\le
P_{\rm out}
=\frac2D\int_0^{D/2}P_{\rm out}(y)\,dy
\le
P_{\rm out}(D/2),
\end{equation*}
because $c(y)$ decreases with $|y|$, and hence
$P_{\rm out}(y)$ is nondecreasing in $|y|$.

\section{Numerical Results}\label{sec:num}
The setting of~\cite{Ref_Ding2025PinchingPerspective} is used: $f_c=28$\,GHz,
$n_\eff=1.4$, $d=3$\,m, $D=10$\,m, $\sigma_w^2=-90$\,dBm and
$\eta=c^2/(4\pi f_c)^2$. The pinches sit above the user, so
$S=2\pi n_\eff/\lambda=0.821$\,rad/mm by \eqref{eq:dphi}, and
$\sigma_\phi=S\sigma_x$ in every closed form; the $N$ pinches of TDMA are
spaced by about $\lambda_g$ at positions that align their phases, as in
Section~\ref{sec:impact}. Under NOMA, $\alpha_s=0.2$ and $\alpha_w=0.8$. The
Monte Carlo simulations use the exact geometry: each pinch is displaced by
$\Delta x_n\sim\mathcal N(0,\sigma_x^2)$, its phase and gain are recomputed
from \eqref{eq:phi} and \eqref{eq:geff}, and $10^6$ draws are taken per point.
For displacements up to $3$\,mm, the exact phase error differs from
$S\Delta x$ by less than $10^{-3}$\,rad and the path gain changes by less than
$2\times10^{-6}$, which confirms Lemma~\ref{lem:S}.

\begin{figure}[!t]
\centering
\includegraphics[width=\columnwidth]{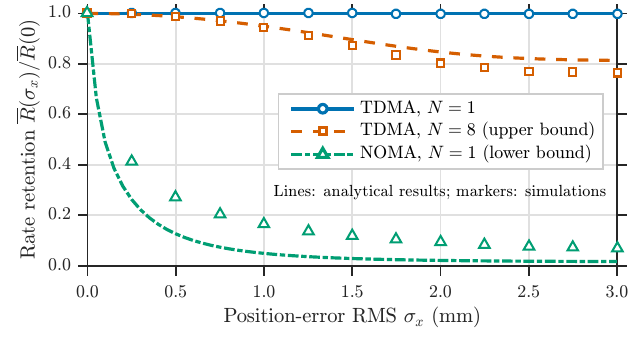}
\caption{Rate retention versus the position error. Lines:
Corollary~\ref{cor:immune} for single-pinch TDMA, the upper bound
\eqref{eq:aN} for eight-pinch TDMA, and the lower bound \eqref{eq:floor} for
the NOMA strong user; markers: exact-geometry Monte Carlo.}
\label{fig:retention}
\end{figure}

We first evaluate the impact of the position error on three schemes: TDMA
with $N=1$, TDMA with $N=8$, and the NOMA strong user. To compare their
sensitivity on one scale, Fig.~\ref{fig:retention} plots the rate retention
$\bar R(\sigma_x)/\bar R(0)$, i.e., the ergodic rate with position error
divided by the error-free rate. At $P=20$\,dBm and $y=0$, where the
single-pinch SNR is $39$\,dB, the error-free rates are $12.98$, $15.98$ and
$10.66$\,bit/s/Hz for the three schemes. Each curve therefore begins at one,
the error-free case, and falls as the position error grows, except for
single-pinch TDMA, which stays at one, as Corollary~\ref{cor:immune} states
and the exact-geometry simulation confirms; the eight-pinch marker at
$\sigma_x=2$\,mm, for example, reads $0.80$: the position error has removed
$20\%$ of the rate. The lines
divide the bounds \eqref{eq:aN} and \eqref{eq:floor} by the error-free rate,
which keeps them upper and lower bounds on the retention. Eight-pinch TDMA
retains $99.7\%$ at $\sigma_x=0.25$\,mm and $94\%$
at $1$\,mm. The bound \eqref{eq:aN} is within $0.005$ of the simulation up to
$1$\,mm and tends to $0.81$, the ratio of the single-pinch rate to the
eight-pinch rate: the mean SNR loses exactly the array gain, and the simulated
rate settles $0.05$ lower because the randomly phased pinches then form a
fading channel. NOMA retains $41\%$ at $0.25$\,mm and $16\%$ at $1$\,mm: a
position error of $2\%$ of the wavelength cuts its rate by more than half. The
bound \eqref{eq:floor} lies below the simulation, $0.05$ against $0.16$ at
$1$\,mm, but predicts the collapse. This ordering holds in the
interference-limited regime of the setting: the NOMA loss is a residual
interference that grows with $P$, whereas the array-gain loss of TDMA is
noise-limited and shrinks with $P$, as \eqref{eq:aN} shows for
$\gamma_0\to\infty$.

\begin{figure}[!t]
\centering
\includegraphics[width=\columnwidth]{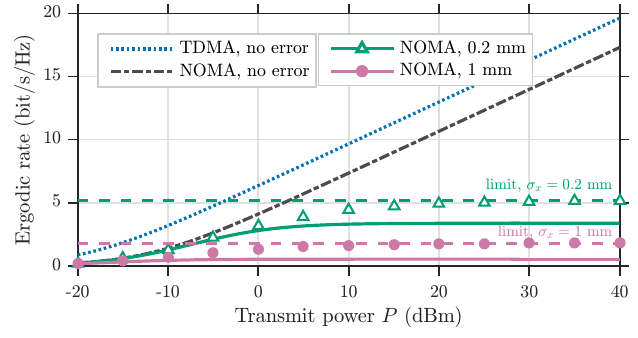}
\caption{Ergodic rate versus transmit power for the NOMA strong user with
one pinch, at $\sigma_x=0.2$ and $1$\,mm. Solid lines: lower bound
\eqref{eq:floor}; dashed lines: $P\to\infty$ limit; markers: exact-geometry
Monte Carlo. Dotted and dash-dot lines: error-free TDMA and NOMA.}
\label{fig:ratevsP}
\end{figure}

Fig.~\ref{fig:ratevsP} shows how the position error limits the NOMA strong
user as the transmit power grows. Without error, the rates of TDMA and NOMA
increase steadily with $P$. With position error, the simulated NOMA rate
first follows the error-free curve and then saturates: as $P$ grows, the
noise term in the SINR \eqref{eq:noma} becomes negligible and the SINR tends
to $\alpha_s/(\alpha_w\Omega)$, which depends only on the residual $\Omega$
of the imperfect SIC, not on $P$. The resulting ceiling,
$\E[\log_2(1+\alpha_s/(\alpha_w\Omega))]$, evaluated numerically, is
$5.19$\,bit/s/Hz at $\sigma_x=0.2$\,mm and $1.79$\,bit/s/Hz at $1$\,mm. The
lower bound \eqref{eq:floor} follows the simulation closely below about
$-15$\,dBm, where noise dominates, and then flattens at $3.37$ and
$0.52$\,bit/s/Hz, below the true ceiling but of the same nature. The ceiling
is thus set by the phase-error spread $\sigma_\phi$ and the power split, and
no increase of the transmit power lifts it.

\begin{figure}[!t]
\centering
\includegraphics[width=\columnwidth]{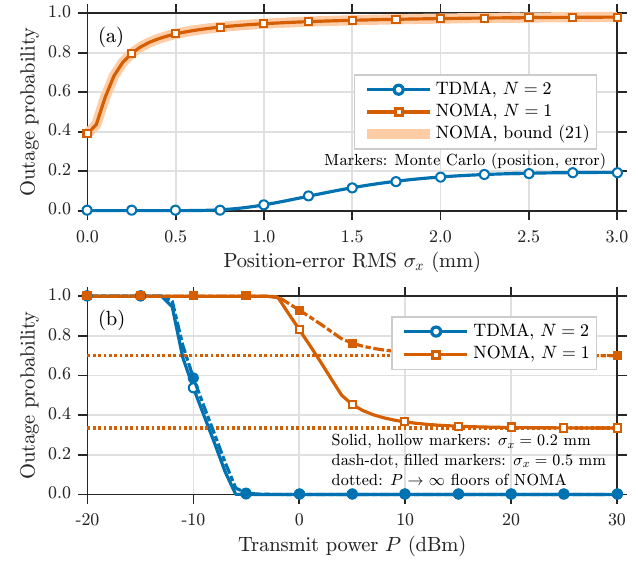}
\caption{Outage probability (a) versus the position error and (b) versus
transmit power at $\sigma_x=0.2$\,mm (solid, hollow markers) and $0.5$\,mm
(dash-dot, filled markers). Lines: \eqref{eq:pout_tdma2} and
\eqref{eq:pout_cond} averaged over $y$; halo in (a): bound
\eqref{eq:pout_avg}; dotted lines in (b): NOMA floors; markers:
exact-geometry Monte Carlo.}
\label{fig:outage}
\end{figure}

Fig.~\ref{fig:outage} shows the spatially averaged outage probability for a
user uniform in $y$, at the target $\SINR_{\rm th}=10$\,dB. The averages of
\eqref{eq:pout_tdma2} and \eqref{eq:pout_cond}, with the series taken to
$m=10$, match the simulation within $0.002$ over the whole range, which
confirms the phase-error model of Section~\ref{sec:chain}. In
Fig.~\ref{fig:outage}(a), at $P=1$\,dBm, position error raises the NOMA
outage steeply, from $0.39$ under ideal placement to $0.80$ at $0.25$\,mm,
whereas two-pinch TDMA is barely affected below $0.5$\,mm and reaches $0.19$
at $3$\,mm; single-pinch TDMA has no outage at this power and is not shown.
The bound \eqref{eq:pout_avg} nearly coincides with the NOMA curve, within
$0.0014$, because its only relaxation, $\Omega\le\Delta\phi_{1s}^2$, is
nearly exact when the tolerance $c(0)$ is small, as it is at this power. In
Fig.~\ref{fig:outage}(b), the TDMA outage falls to zero as the power grows,
whereas the NOMA outage settles on the floors $0.34$ and $0.70$ at
$\sigma_x=0.2$ and $0.5$\,mm, the limit of \eqref{eq:pout_cond} as
$P\to\infty$, which no increase of $P$ removes.

\section{Conclusion}\label{sec:concl}
This paper modeled the position errors of pinching antennas and traced them
to the link: each displacement maps linearly to a phase error, whose rms
value and coherence factor set the ergodic rate and the outage probability
of TDMA and NOMA. Single-pinch TDMA is immune, multi-pinch TDMA loses at most
its array gain, and NOMA meets a rate ceiling and an outage floor that no
increase of transmit power removes: the outcome is set by how many phasors
must stay aligned, not by the hardware tolerance alone. Simulations with the
exact geometry confirm the closed forms and show that, at $28$\,GHz, a
position error of $0.2$\,mm already halves the NOMA rate at high SNR. The
analysis assumes a single waveguide and a line-of-sight channel; the same
chain extends to multipath and to space division over several waveguides,
left for future work.

\bibliographystyle{IEEEtran}
\bibliography{main}

\end{document}

%% file: fig_sysmodel.tex
\begin{tikzpicture}[>=Stealth,font=\sffamily\small,line width=0.5pt]
  \shade[top color=black!8, bottom color=black!2] (0,0) -- (7.9,0) -- (9.37,0.85) -- (1.47,0.85) -- cycle;
  \draw[black!25, line width=0.3pt] (0,0) -- (7.9,0) -- (9.37,0.85) -- (1.47,0.85) -- cycle;
  \draw[black!10, line width=0.3pt] (2,0) -- (3.47,0.85) (4,0) -- (5.47,0.85) (6,0) -- (7.47,0.85);
  \draw[black!10, line width=0.3pt] (0.49,0.283) -- (8.39,0.283) (0.98,0.567) -- (8.88,0.567);
  \draw[->,black!65,thick] (0,0) -- (9.3,0) node[right,text=black] {$x$};
  \draw[->,black!65,thick] (0,0) -- (1.66,0.96) node[above right,text=black] {$y$};
  \draw[->,black!65,thick] (0,0) -- (0,3.75) node[above,text=black] {$z$};
  \node[anchor=north east,inner sep=1.5pt,text=black] at (0,0) {$O$};
  \shade[top color=blue!24, bottom color=blue!6] (1.9,2.51) rectangle (8.3,2.69);
  \draw[blue!55!black,line width=0.55pt] (1.9,2.51) -- (8.3,2.51) (1.9,2.69) -- (8.3,2.69);
  \draw[blue!55!black,fill=blue!16,line width=0.55pt] (8.3,2.6) ellipse (0.06 and 0.09);
  \draw[blue!55!black,fill=blue!16,line width=0.55pt] (1.9,2.6) ellipse (0.06 and 0.09);
  \node[blue!45!black,anchor=north east,font=\sffamily\footnotesize] at (8.25,2.46) {waveguide};
  \filldraw[fill=blue!18,draw=blue!65!black,rounded corners=1.5pt] (0.85,2.40) rectangle (1.70,2.80);
  \node[font=\sffamily\footnotesize,text=blue!65!black] at (1.275,2.6) {feed};
  \draw[->,black!70] (0.35,2.6) -- (0.82,2.6);
  \node[anchor=south,font=\sffamily\footnotesize,text=black] at (0.58,2.65) {$s$};
  \draw[blue!65!black,line width=0.7pt] (1.70,2.6) -- (1.86,2.6);
  \node[anchor=north,font=\sffamily\footnotesize,text=black] at (1.9,2.22) {$\boldsymbol{\psi}_0$};
  \draw[dashed,black!45] (1.9,2.28) -- (1.9,2.50);
  \foreach \px in {3.0} {
    \shade[top color=blue!40!black, bottom color=blue!75!black] (\px-0.13,2.69) -- (\px+0.13,2.69) -- (\px+0.075,2.86) -- (\px-0.075,2.86) -- cycle;
    \shade[top color=blue!75!black, bottom color=blue!40!black] (\px-0.13,2.51) -- (\px+0.13,2.51) -- (\px+0.075,2.34) -- (\px-0.075,2.34) -- cycle;
    \draw[blue!60!black,line width=0.4pt] (\px-0.13,2.69) -- (\px+0.13,2.69) -- (\px+0.075,2.86) -- (\px-0.075,2.86) -- cycle;
    \draw[blue!60!black,line width=0.4pt] (\px-0.13,2.51) -- (\px+0.13,2.51) -- (\px+0.075,2.34) -- (\px-0.075,2.34) -- cycle;
    \draw[red!65!black,line width=0.65pt] ($(\px,2.29)+(210:0.13)$) arc (210:330:0.13);
    \draw[red!65!black!55,line width=0.5pt] ($(\px,2.29)+(210:0.21)$) arc (210:330:0.21);
  }
  \node[anchor=south,text=black,inner sep=1pt] at (3.0,2.90) {$\boldsymbol{\psi}_1$};
  \foreach \px in {5.35} {
    \draw[black!45,dashed,line width=0.4pt] (\px-0.13,2.69) -- (\px+0.13,2.69) -- (\px+0.075,2.86) -- (\px-0.075,2.86) -- cycle;
    \draw[black!45,dashed,line width=0.4pt] (\px-0.13,2.51) -- (\px+0.13,2.51) -- (\px+0.075,2.34) -- (\px-0.075,2.34) -- cycle;
  }
  \node[anchor=east,font=\sffamily\footnotesize,text=black,align=right,inner sep=1pt] at (5.22,3.02) {nominal $x_n$};
  \foreach \px in {5.9} {
    \shade[top color=blue!40!black, bottom color=blue!75!black] (\px-0.13,2.69) -- (\px+0.13,2.69) -- (\px+0.075,2.86) -- (\px-0.075,2.86) -- cycle;
    \shade[top color=blue!75!black, bottom color=blue!40!black] (\px-0.13,2.51) -- (\px+0.13,2.51) -- (\px+0.075,2.34) -- (\px-0.075,2.34) -- cycle;
    \draw[blue!60!black,line width=0.4pt] (\px-0.13,2.69) -- (\px+0.13,2.69) -- (\px+0.075,2.86) -- (\px-0.075,2.86) -- cycle;
    \draw[blue!60!black,line width=0.4pt] (\px-0.13,2.51) -- (\px+0.13,2.51) -- (\px+0.075,2.34) -- (\px-0.075,2.34) -- cycle;
    \draw[red!65!black,line width=0.65pt] ($(\px,2.29)+(210:0.13)$) arc (210:330:0.13);
    \draw[red!65!black!55,line width=0.5pt] ($(\px,2.29)+(210:0.21)$) arc (210:330:0.21);
  }
  \node[anchor=south west,text=black,inner sep=1pt] at (6.0,2.88) {$(x_n{+}\Delta x_n,0,d)$};
  \draw[->,red!70!black,line width=0.8pt] (5.35,3.04) -- (5.9,3.04);
  \draw[dashed,black!45] (5.9,3.14) -- (5.9,2.88);
  \node[anchor=south,font=\sffamily\footnotesize,text=red!70!black] at (5.62,3.06) {$\Delta x_n$};
  \draw[<->,black!60] (1.9,3.42) -- (5.35,3.42);
  \draw[dashed,black!45] (1.9,3.42) -- (1.9,2.72);
  \draw[dashed,black!45] (5.35,3.42) -- (5.35,2.88);
  \node[anchor=south,font=\sffamily\footnotesize,text=black] at (3.6,3.44) {$d_n$ (guided, grows by $\Delta x_n$)};
  \foreach \px/\py/\nm in {4.42/0.30/uk, 6.65/0.09/ukp} {
    \draw[fill=white,rounded corners=1.6pt,line width=0.55pt,
          drop shadow={shadow xshift=0.8pt, shadow yshift=-0.8pt, fill=black, opacity=0.18}]
          (\px-0.18,\py) rectangle (\px+0.18,\py+0.62);
    \shade[top color=blue!22, bottom color=blue!6] (\px-0.135,\py+0.105) rectangle (\px+0.135,\py+0.50);
    \draw[line width=0.45pt,black!70] (\px-0.055,\py+0.555) -- (\px+0.055,\py+0.555);
    \fill[black!60] (\px,\py+0.052) circle (0.022);
    \coordinate (\nm) at (\px,\py+0.62);
  }
  \draw[dashed,black!45] (4.42,0.30) -- (3.9,0);
  \fill[black!50] (3.9,0) circle (0.025);
  \node[anchor=north,font=\sffamily\footnotesize,text=black] at (3.9,-0.05) {$x_k$};
  \node[anchor=north west,font=\sffamily\footnotesize,text=black,inner sep=1pt] at (4.12,0.20) {$y_k$};
  \node[anchor=west,font=\sffamily\footnotesize,text=black,inner sep=1pt] at (4.66,0.62) {$\mathbf{u}_k$};
  \node[anchor=west,font=\sffamily\footnotesize,text=black,inner sep=1pt] at (6.89,0.42) {$\mathbf{u}_{k'}$};
  \node[anchor=north east,font=\sffamily\footnotesize,text=black,inner sep=1.5pt] at (3.50,-0.05) {feed side, $\cos\theta{>}0$};
  \node[anchor=north,font=\sffamily\footnotesize,text=black,inner sep=1.5pt] at (6.60,-0.05) {far side, $\cos\theta{<}0$};
  \coordinate (pn) at (5.9,2.26);
  \draw[dashed,black!55] (pn) -- (uk);
  \node[anchor=west,font=\sffamily\footnotesize,text=black,inner sep=1pt] at (5.30,1.62) {$r_{nk}$};
  \draw[dashed,black!35] (pn) -- (ukp);
  \draw[dashed,black!35] (3.0,2.26) -- (uk);
  \node[anchor=east,font=\sffamily\footnotesize,text=black,inner sep=1pt] at (3.60,1.62) {$r_{1k}$};
  \draw[black!60,line width=0.4pt] (pn) -- (3.9,2.26);
  \draw[black!60,line width=0.4pt] (3.9,2.20) -- (3.9,2.32);
  \draw[dashed,black!40,line width=0.3pt] (3.9,2.20) -- (3.9,0.04);
  \node[anchor=north,font=\sffamily\footnotesize,text=black,inner sep=1.5pt]
       at (4.45,2.22) {$x_n-x_k$};
  \draw[->,black!70,line width=0.45pt] ($(pn)+(180:0.5)$) arc (180:222:0.5);
  \node[font=\sffamily\footnotesize,text=black,inner sep=1pt] at ($(pn)+(203:0.72)$) {$\theta$};
  \draw[dashed,black!40] (8.3,2.6) -- (8.75,2.6);
  \draw[<->,black!60] (8.75,0) -- (8.75,2.6);
  \node[anchor=west,font=\sffamily\footnotesize,text=black] at (8.78,1.30) {$d$};
\end{tikzpicture}

%% file: main.bbl
\begin{thebibliography}{10}
\providecommand{\url}[1]{#1}
\csname url@samestyle\endcsname
\providecommand{\newblock}{\relax}
\providecommand{\bibinfo}[2]{#2}
\providecommand{\BIBentrySTDinterwordspacing}{\spaceskip=0pt\relax}
\providecommand{\BIBentryALTinterwordstretchfactor}{4}
\providecommand{\BIBentryALTinterwordspacing}{\spaceskip=\fontdimen2\font plus
\BIBentryALTinterwordstretchfactor\fontdimen3\font minus
  \fontdimen4\font\relax}
\providecommand{\BIBforeignlanguage}[2]{{%
\expandafter\ifx\csname l@#1\endcsname\relax
\typeout{** WARNING: IEEEtran.bst: No hyphenation pattern has been}%
\typeout{** loaded for the language `#1'. Using the pattern for}%
\typeout{** the default language instead.}%
\else
\language=\csname l@#1\endcsname
\fi
#2}}
\providecommand{\BIBdecl}{\relax}
\BIBdecl

\bibitem{Ref_Ding2025PinchingPerspective}
Z.~Ding, R.~Schober, and H.~V. Poor, ``Flexible-antenna systems: A
  pinching-antenna perspective,'' \emph{IEEE Trans. Commun.}, vol.~73, no.~10,
  pp. 9236--9253, Oct. 2025.

\bibitem{Ref_Liu2025Tutorial}
Y.~Liu \emph{et~al.}, ``Pinching-antenna systems ({PASS}): A tutorial,''
  \emph{IEEE Trans. Commun.}, vol.~74, pp. 4881--4918, 2026.

\bibitem{Ref_Tyrovolas2026Performance}
D.~Tyrovolas, S.~A. Tegos, P.~D. Diamantoulakis, S.~Ioannidis, C.~K. Liaskos,
  and G.~K. Karagiannidis, ``Performance analysis of pinching-antenna
  systems,'' \emph{IEEE Trans. Cogn. Commun. Netw.}, vol.~12, pp. 590--601,
  2026.

\bibitem{Ref_Yue2026NOMAPerformance}
X.~Yue, X.~Tao, J.~Zhao, X.~Lei, Y.~Liu, and Z.~Ding, ``Performance analysis of
  pinching antenna systems enabled {NOMA} communications,''
  \emph{arXiv:2604.25285}, 2026.

\bibitem{Ref_Pakravan2026SecrecyPosition}
S.~Pakravan, I.~Trigui, W.~Ajib, and W.-P. Zhu, ``Impact of position
  uncertainty on the secrecy performance of pinching-antenna systems,''
  \emph{IEEE Commun. Lett.}, vol.~30, pp. 2765--2769, 2026.

\bibitem{Ref_Chen2026HybridPASS}
K.~Chen, C.~Qi, O.~A. Dobre, and C.~Yuen, ``Hybrid pinching antenna systems:
  Architecture and beamforming design,'' \emph{IEEE Trans. Wireless Commun.},
  vol.~25, pp. 12\,129--12\,144, 2026.

\bibitem{Ref_AlaaEldin2026BERNOMA}
M.~AlaaEldin, A.~S. Inwood, X.~Mu, and M.~Matthaiou, ``{BER} analysis and
  optimization of pinching-antenna-based {NOMA} communications,'' in
  \emph{Proc. IEEE Int. Conf. Commun. (ICC)}, Glasgow, U.K., May 2026, pp.
  1--6.

\bibitem{Ref_Wang2026RSMAMultiCarrier}
P.~Wang, H.~Wang, Y.~Fu, and R.~Song, ``Optimization for pinching antennas
  system with multiple carriers and rate splitting multiple access,''
  \emph{arXiv:2604.14736}, 2026.

\bibitem{Ref_Bozkurt2026Trajectory}
Y.~{Tokur Bozkurt}, ``Trajectory-aware antenna reconfiguration for pinching
  antenna systems: Performance analysis under user mobility,'' \emph{Digit.
  Signal Process.}, vol. 178, p. 106162, 2026.

\bibitem{Ref_Wang2025AntennaActivationNOMA}
K.~Wang, Z.~Ding, and R.~Schober, ``Antenna activation for {NOMA} assisted
  pinching-antenna systems,'' \emph{IEEE Wireless Commun. Lett.}, vol.~14,
  no.~5, pp. 1526--1530, May 2025.

\end{thebibliography}
